\documentclass[sigconf,screen,nonacm]{acmart}
\usepackage{algorithm}
\usepackage{algpseudocode}
\usepackage{booktabs}
\usepackage{graphicx}
\usepackage{url}
\usepackage{tabularx}

\newtheorem{claim}{Claim}

\begin{document}


\title{Predictable LLM Serving on GPU Clusters}

\author{Erfan Darzi}
\affiliation{%
  \institution{Harvard University, MIT}
  \city{Cambridge}
  \state{MA}
  \country{USA}
}

\author{Shreeanant Bharadwaj}
\affiliation{%
  \institution{Northeastern University}
  \city{Boston}
  \state{MA}
  \country{USA}
}

\author{Sree Bhargavi Balija}
\affiliation{%
  \institution{University of California, San Diego}
  \city{La Jolla}
  \state{CA}
  \country{USA}
}



\begin{abstract}
Latency-sensitive inference on shared A100 clusters often suffers noisy-neighbor interference on the PCIe fabric, inflating tail latency and SLO violations. We present a fabric-agnostic, VM-deployable host-level controller that combines dynamic Multi-Instance GPU (MIG) reconfiguration, PCIe-aware placement, and lightweight guardrails (MPS quotas, cgroup I/O). It samples per-tenant tails and system signals, uses topology hints to avoid PCIe hot spots, and gates actions with dwell/cool-down to avoid thrash. On a single host and a 2-node (16-GPU) cluster, SLO miss-rate is reduced by \(\approx\)32\% (\(\approx\)1.5×) and p99 latency improves \(\approx\)15\% with \(\leq\)5\% throughput cost versus static MIG and naive placement; ablations show MIG and placement contribute comparably. We also evaluate LLM serving with vLLM on OLMo 2 7B Instruct: TTFT p99 improves \(\approx\)10--15\% at \(\leq\)5\% cost without changing the controller. 
\end{abstract}
\maketitle

\section{Introduction}

For latency-sensitive GPU inference services, controlling tail latency is paramount. On shared hosts, the "noisy-neighbor" problem is a primary cause of unpredictable performance, where co-located tenants---such as batch training jobs, data loaders, or background I/O-intensive tasks---create resource contention that leads to significant jitter in inference latency. This jitter frequently causes violations of Service Level Objectives (SLOs), exhausting error budgets and degrading service quality. A key source of this interference is contention for shared Peripheral Component Interconnect Express (PCIe) bandwidth, which becomes a bottleneck during intensive data transfers between the host and GPU~\cite{tang2025pcie, li2019priority}. Our scope includes LLM serving with vLLM; we additionally track time-to-first-token (TTFT) tails.
\begin{figure*}[t]
    \centering
\includegraphics[width=0.99\textwidth]{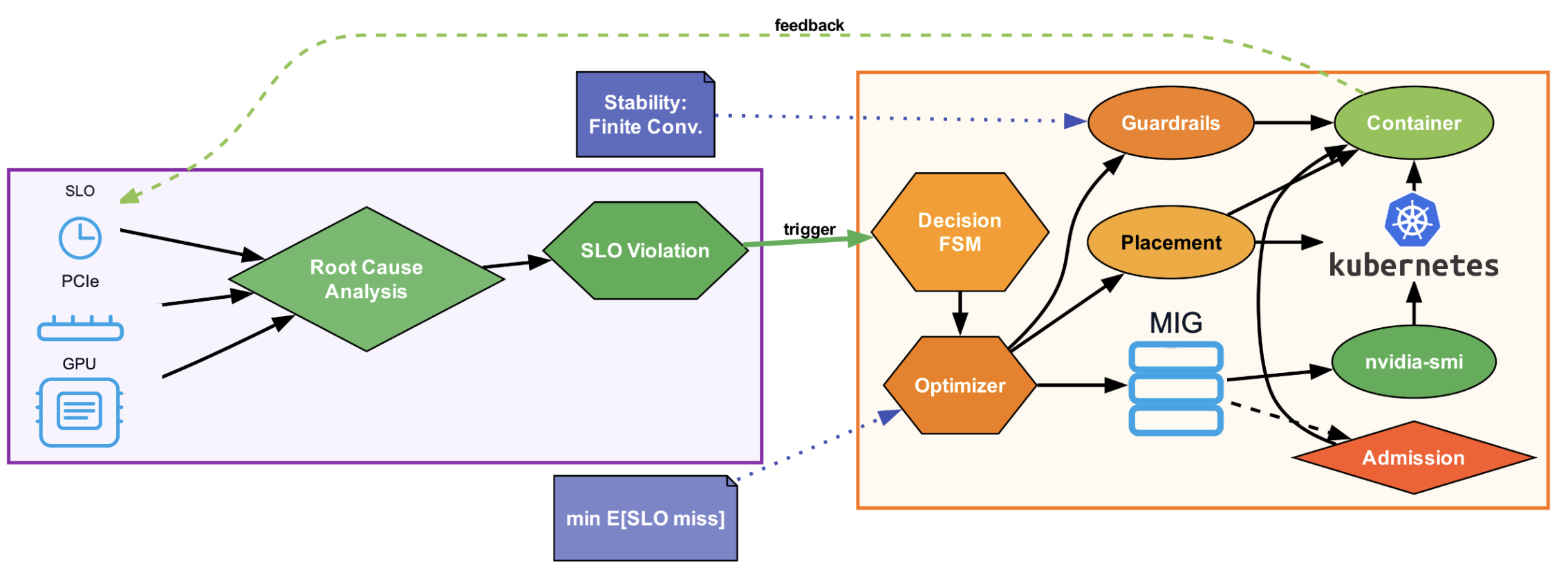}
    \caption{System architecture of the multi-tenancy controller. The monitoring domain detects SLO violations and performs root-cause analysis; the decision FSM and optimizer select among dynamic MIG reconfiguration, PCIe-aware placement, and lightweight guardrails; the execution path applies changes via \texttt{nvidia-smi} and runtime controls (MIG UUID selection, process pinning, cgroups), with a feedback loop to latency tracking.}
    \label{fig:mig_controller}
\end{figure*}
This challenge is particularly acute in multi-tenant environments with limited visibility into the underlying hardware topology, including PCIe lane allocation and Non-Uniform Memory Access (NUMA) domains~\cite{jeon2019analysis}. In practice, MIG provides hard isolation for compute and HBM but not for the shared PCIe path, so co-running PCIe-bound tenants can still throttle one another across MIG instances~\cite{tang2025pcie}. Meanwhile, pure MPS-based sharing improves utilization but offers weaker isolation for memory and I/O~\cite{wu2023transparent}. Consequently, there is a strong need for host-only solutions that mitigate interference and guarantee performance while operating within the constraints of a renter's VM.

We target this setting with a controller that integrates three simple levers: (i) dynamic MIG to increase or relax isolation as conditions change, (ii) PCIe-aware placement to avoid topologically hot paths, and (iii) lightweight guardrails (MPS quotas and cgroup I/O throttles) to contain bursty neighbors. The controller is deliberately conservative: it only escalates isolation when p99/p999 tails persist above threshold, enforces dwell/cool-down to avoid thrashing, and returns resources when the system is stable. Our approach complements recent work showing the benefits of dynamic MIG reconfiguration~\cite{wang2024migrator, li2022miso} and PCIe-aware scheduling~\cite{tang2025pcie}, but is deployable by renters without fabric privileges.

Our contributions are:
\begin{enumerate}
    \item A multi-tenancy controller for NVIDIA A100 GPUs that integrates dynamic MIG and PCIe-aware placement.
    \item First SLO-safe, multi-tenant control demo on a multi-node (16-GPU) cloud cluster without fabric privileges.
\end{enumerate}

\section{Method}
We implement a controller that minimizes tail latency for GPU workloads. The controller operates in a continuous loop, sampling per-tenant tail latency and system-level signals such as PCIe bandwidth and I/O activity. If a tenant's latency exceeds a predefined SLO threshold, the controller acts from a three-tiered decision space to mitigate interference. It can (1) dynamically reconfigure MIG profiles to provide stronger hardware isolation, (2) perform PCIe-aware placement to move tenants away from hardware hot spots, or (3) apply lightweight guardrails such as MPS quotas and I/O throttles. To prevent oscillation, the policy enforces dwell times and cool-down periods between actions. This approach allows for adaptive, fine-grained resource management to ensure SLO compliance under dynamic load.

\paragraph{Signals and sampling.} Every $\Delta$ seconds (1--5s), the controller collects: (i) per-tenant p95/p99/p999 latencies, including TTFT for autoregressive serving, and SLO miss-rate; (ii) NVML metrics (SM utilization, memory bandwidth); (iii) PCIe counters (bytes/s; retries if available) and host block I/O; and optionally, (iv) NIC and IRQ activity. Signals are smoothed with exponential moving averages and hysteresis to reduce spurious triggers.

\paragraph{Objective and constraints.} The primary objective is to minimize SLO miss-rate and p99/p999 latency for the latency-sensitive tenant subject to a throughput budget (\(\leq\)5\% degradation). Actions must be infrequent enough to keep reconfiguration overhead small (\(\leq\)30s per change on A100), so we impose a minimum dwell time between changes and a cool-down period after recovery.

\subsection{Signals \& Objective}
Our controller's decision-making process is driven by a continuous stream of performance signals, polled at a configurable interval, $\Delta$ (typically 1--5 seconds), using standard monitoring tools like the NVIDIA Management Library (NVML). The primary signal is the per-tenant 99th percentile (p99) latency, which directly measures adherence to its SLO. To diagnose the root cause of SLO violations, the controller also ingests a set of secondary signals. These include Peripheral Component Interconnect Express (PCIe) counters (bytes/s, retries) to detect bandwidth contention, NVML metrics for Streaming Multiprocessor (SM) and memory bandwidth utilization to identify compute or memory pressure, and host-level I/O statistics to correlate with storage-intensive noisy neighbors. Optionally, Network Interface Controller (NIC) and Interrupt Request (IRQ) statistics can be used to identify network or host-level interference. The controller's objective is to minimize the SLO miss-rate and p99 latency for the primary tenant, with a throughput degradation constraint of 5\% or less.

\subsection{Decision Space}
To counteract the sources of interference identified by its signals, the controller operates within a well-defined decision space encompassing three categories of actions. First, it can perform \textbf{Dynamic MIG Reconfiguration}~\cite{wang2024migrator, li2022miso}, changing the size of a tenant's Multi-Instance GPU (MIG) profile, such as `1g.10gb` to `2g.20gb` on an A100. This provides stronger hardware isolation for compute and memory resources. Second, it can perform \textbf{PCIe-Aware Placement}, migrating a tenant between different MIG instances on the same GPU. This action is guided by topology-aware heuristics that seek to map tenants to MIG instances that do not share a busy PCIe root switch or Non-Uniform Memory Access (NUMA) node with a high-traffic neighbor~\cite{tang2025pcie}. Finally, the controller can apply lightweight, host-level \textbf{Guardrails}. These include adjusting the Multi-Process Service (MPS) active thread percentage to cap a tenant's concurrency and applying cgroup I/O throttles using `io.max` to limit the disk bandwidth of a noisy background job~\cite{wu2023transparent}.

\subsubsection{Topology-aware placement heuristic}
We query GPU and PCIe topology via DCGM/NVML and host tools including `lspci` and NUMA maps to form a simple placement score for each candidate MIG instance. The score penalizes (i) sharing a PCIe root complex with a bandwidth-heavy tenant, (ii) colocating with a NUMA domain exhibiting high block I/O, and (iii) recent IRQ bursts on adjacent CPU cores. When upgrading isolation, we first attempt an intra-GPU move to the least-penalized MIG instance; only if insufficient do we enlarge the MIG slice. When relaxing isolation, we select a smaller profile whose placement score remains below a conservative threshold, preventing regressions.

\subsection{Policy}
The controller's policy is a simple, robust state machine designed to react to performance degradation while avoiding oscillation. The policy is triggered when a tenant's p99 latency exceeds a threshold, $\tau$, for $Y$ consecutive observation windows. Upon triggering, the controller first attempts to diagnose the issue using its secondary signals. If high PCIe or I/O pressure is detected, it applies a cgroup I/O throttle for a fixed duration, $Z$, on the offending background tenant. If the primary cause appears to be compute or memory contention, or if throttling does not resolve the issue, the controller proceeds to \textbf{upgrade the tenant's isolation}. This involves either increasing its MIG share (if headroom is available on the GPU) or migrating it to a less contended MIG instance. As part of this upgrade, it also pins the tenant's CPU affinity away from cores handling high IRQ loads and enforces a stricter MPS quota. Conversely, when a tenant's performance has been stable and well within its SLO for a sustained period, the controller may attempt to \textbf{relax its isolation}, moving it to a smaller MIG profile to free up resources and improve overall GPU utilization. To prevent rapid, counter-productive changes, the policy enforces a minimum dwell time between actions and a cool-down period before re-evaluating a tenant's status after a move. In cases where no safe placement can be found for a new tenant without violating the SLOs of existing tenants, an admission control mechanism will queue or reject the new workload.

\subsection{Implementation Notes}
The controller is implemented in Python and deployed as a host-level service. It uses \texttt{nvidia-smi mig} for reconfiguration and enforces placement via process pinning with \texttt{CUDA\_VISIBLE\_DEVICES} and cgroups. Root access within the VM is sufficient. We gate actions behind feature flags, log all decisions with signal snapshots for audit, and support rollback to the last-known-good configuration if post-change p99 worsens within a short validation window. CPU affinity is applied to keep the latency-sensitive tenant away from IRQ-heavy cores; MPS quotas are adjusted via \texttt{CUDA\_MPS\_ACTIVE\_THREAD\_PERCENTAGE}; and I/O throttles use cgroup \texttt{io.max} with bounded windows (tens of seconds) to reduce collateral damage.

\begin{algorithm}
\caption{Multi-Tenancy Controller}
\begin{algorithmic}[1]
\State \textbf{Inputs:} Latency stream $\mathcal{L}$, tail threshold $\tau$, persistence $Y$
\State \textbf{State:} Window $W$, config $C$, cooldown timer $T_{cd}$
\Procedure{OnObservation}{$l$}
    \State $W \gets (W \cup \{l\})$
    \State $p_{99} \gets \text{quantile}(W, 0.99)$
    \If{not at\_reconfig\_boundary() or is\_cooling\_down()} \textbf{return} \EndIf
    \If{$p_{99} > \tau$ for $Y$ consecutive windows}
        \State $C \gets \text{UpgradeIsolation}(C)$
        \State Relaunch(C); $T_{cd} \gets \text{DWELL\_TIME}$
    \ElsIf{tail\_is\_stable() and throughput\_ok()}
        \State $C \gets \text{RelaxIsolation}(C)$
        \State Relaunch(C); $T_{cd} \gets \text{DWELL\_TIME}$
    \EndIf
\EndProcedure
\end{algorithmic}
\label{alg:controller}
\end{algorithm}

\subsection{Formal Modeling}

\subsubsection{Modeling PCIe Contention and Latency Interference}
We model the PCIe fabric as a single processor-sharing (PS) server of capacity \(B\). When a set \(\mathcal{A}(t)\) of tenants is active, tenant \(i\) receives instantaneous bandwidth
\[
    b_i(t) = \min\left\{ \frac{B\,w_i}{\sum_{j\in \mathcal{A}(t)} w_j},\; g_i \right\},
\]
where \(w_i>0\) are optional weights (equal weights recover equal sharing) and \(g_i\) is an optional host-level throttle. This captures both equal/weighted sharing and explicit rate limits.

\begin{figure*}[t]
    \centering
    \includegraphics[width=0.7\textwidth]{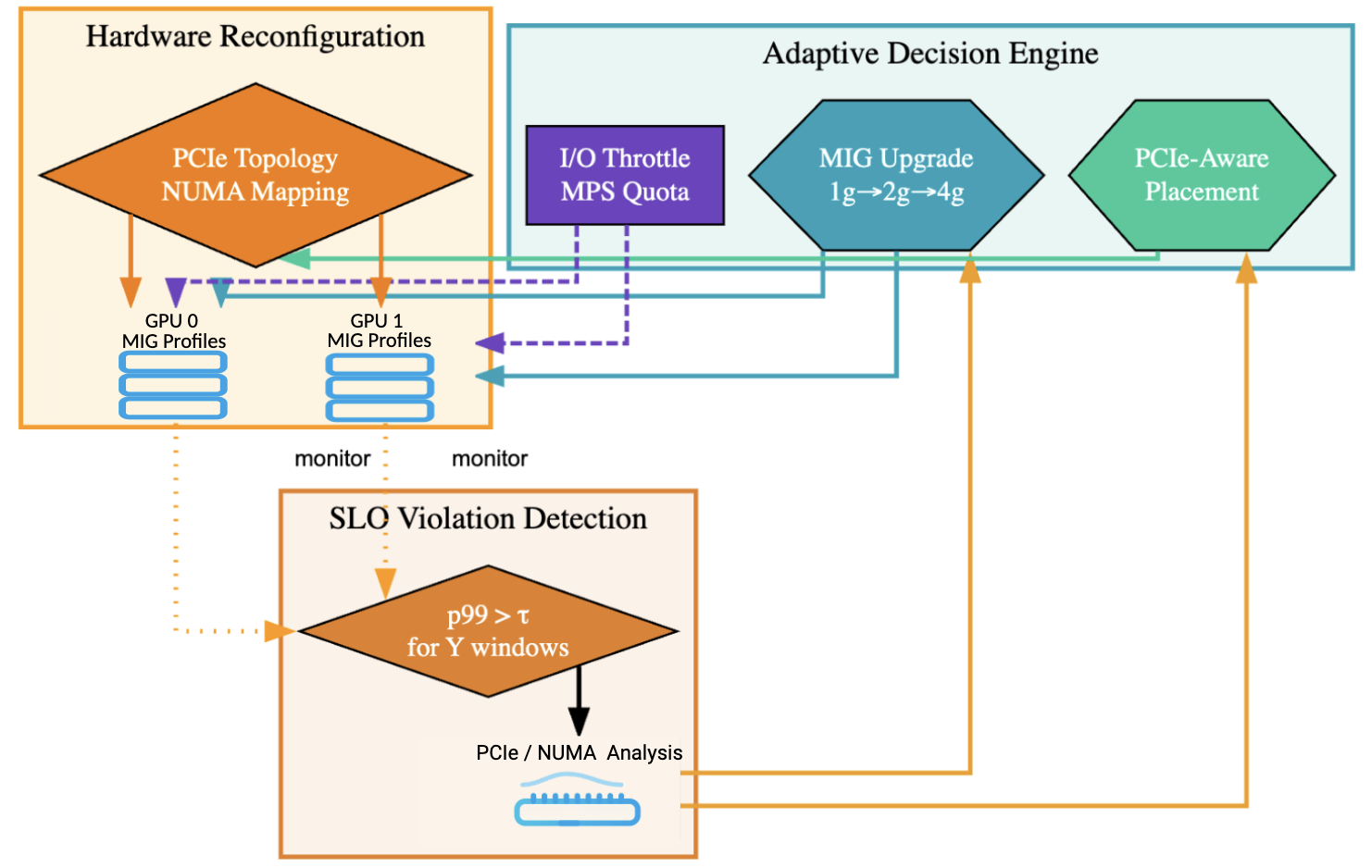}
    \caption{PCIe contention model and topology. Under a PS model with capacity \(B\), co-active tenants share bandwidth; caps \(g_i\) (when applied) prevent pathological oversubscription. Saturation inflates the transfer component of latency and contributes to heavy tails.}
    \label{fig:pcie_contention}
\end{figure*}

The latency \(L_i(t)\) for tenant \(i\)'s inference requests includes a transfer component proportional to the data size \(s_i\) divided by \(b_i(t)\), plus compute time \(c_i\):
\(L_i(t) = c_i + \frac{s_i}{b_i(t)} + \epsilon(t)\), where \(\epsilon(t)\) captures stochastic noise from queuing and scheduling. For tails, we avoid positing a parametric form. Instead, for guidance we use the classical Kingman approximation for the mean queueing delay in a G/G/1 stage upstream of compute:
\[
    \mathbb{E}[W_q] \approx \frac{\rho}{1-\rho}\,\frac{c_a^2 + c_s^2}{2}\,\mathbb{E}[S],\quad \rho = \lambda\,\mathbb{E}[S],
\]
where \(c_a\) and \(c_s\) are the coefficients of variation of inter-arrival and service times, respectively. We report empirical \(p99/p999\) in evaluation and use the bound qualitatively to explain how saturation (\(\rho \to 1\)) inflates tails.

\subsubsection{Optimization Objective for SLO Compliance}
The controller optimizes SLO compliance by minimizing the expected SLO miss-rate \(M = \Pr(L_i > \tau)\) for the primary tenant, subject to throughput constraints. Formally, define the objective as \(\min \mathbb{E}[M]\) over MIG configurations \(m_i \in \mathcal{M}\) and placements \(p_i \in \mathcal{P}\) (PCIe topology-aware slots), with guardrail parameters \(g_i\) (MPS quota, I/O throttle). The constraint ensures throughput \(T_i \geq 0.95 T_{\text{base}}\), where \(T_i = \lambda_i (1 - M)\).

To solve this, we use a greedy heuristic approximating the NP-hard allocation problem. Consider the isolation upgrade: if \(p_{99,i} > \tau\) persists, upgrade \(m_i\) to maximize \(\Delta \mu_i = \mu(m_i') - \mu(m_i)\), where \(\mu(m) \propto\) SM cores and memory in profile \(m\). Finite termination of isolation upgrades: the action/state space is finite, and each upgrade strictly increases isolation (more SMs and/or memory), so the policy performs at most \(|\mathcal{M}| - 1\) upgrades before reaching maximal isolation. We do not claim that this guarantees \(p_{99,i} \le \tau\); it only ensures that upgrade sequences terminate in finitely many steps.

\subsubsection{Theorem on Guardrail Stability}
\begin{claim}[Stability under PS and bounded throttles]
Assume (i) the PCIe fabric behaves as a single PS server of capacity \(B\); (ii) each tenant \(j\)'s effective PCIe demand is bounded by a throttle \(g_j\) when active; (iii) the aggregate offered load satisfies \(\sum_j g_j < B\); and (iv) arrivals to the latency-sensitive tenant \(i\) are stationary with \(\lambda_i < \mu_i(B)\) under these throttles. Then the queue for tenant \(i\) is stable and its latency distribution has finite moments; in particular, \(p_{99,i}\) exists and is finite.
\end{claim}
\begin{proof}[Sketch]
Under (i)–(iii) the PCIe stage is a PS queue with utilization \(\rho = \sum_j g_j / B < 1\). Standard results imply positive recurrence and finite mean waiting time for each flow. With (iv), the end-to-end service rate for tenant \(i\) exceeds its arrival rate. Our controller only increases isolation or tightens throttles when observed tails exceed \(\tau\) and never violates (iii), so it reaches a stable regime in finitely many steps. We report \(p99/p999\) empirically; the model serves as qualitative guidance rather than an exact tail predictor.
\end{proof}

\section{Single- and Multi-Node Evaluation}
\subsection{Setup}
Our multi-node evaluation is conducted on a 2-node cluster consisting of two AWS `p4d.24xlarge` instances, each with eight NVIDIA A100-80GB GPUs, linked by a 200 Gbps InfiniBand/EFA interconnect. The 16-GPU pool is orchestrated by Slurm 23.02, with MOFED 5.9 and DCGM 3.2 installed on both nodes. All tenants run inside Docker containers. Experiments were repeated 7 times with fixed seeds; we report means with 95\% confidence intervals. We pin CPU cores for the SLO-sensitive tenant and isolate background dataloaders to separate NUMA domains when possible. Reconfiguration dwell/cool-down are configured per Table~\ref{tab:parameters}.

\noindent\textbf{Workloads:} We design a suite of three co-located tenants to simulate a realistic multi-tenancy scenario.
\begin{itemize}
    \item \textbf{T1 (Latency-Sensitive Tenant):} Our primary tenant is a latency-sensitive inference workload with a p99 latency SLO of 15 ms in the non-LLM experiments.
    \item \textbf{T2 (Bandwidth-Heavy Tenant):} To create sustained PCIe and memory bandwidth pressure, we use a background tenant that simulates a data-intensive Extract, Transform, Load (ETL) process. This tenant continuously reads large files from NVMe storage into host memory, performs a data transformation kernel on the GPU, and writes the results back. This workload is designed to contend for PCIe bandwidth.
    \item \textbf{T3 (Compute-Heavy Tenant):} To create SM contention, we run a background training job using a synthetic model. This tenant is designed to be compute-bound, maximizing its SM occupancy and creating a "noisy neighbor" that interferes with T1's access to compute resources.
\end{itemize}
An interference script toggles the activity of T2 and T3 to create dynamic periods of contention, allowing us to evaluate the controller's responsiveness. The key parameters for our controller are detailed in Table \ref{tab:parameters}. Unless otherwise noted, the SLO for T1 is 15 ms (p99). For the LLM case study (Section "LLM Serving Case Study"), we use TTFT p99 \(\le\) 200 ms. Batch size is 1, and input sizes are drawn from a realistic mixture to induce time-varying PCIe pressure. We collect NVML/DCGM, host I/O, and controller logs to correlate actions with outcomes.

\subsection{Metrics}
Our primary metric is SLO miss-rate; secondary metrics include p99/p999 latency and throughput. We measure the SLO miss-rate, defined as the percentage of inference requests that exceed the 15ms latency target. We also record the full latency distribution, reporting the 95th, 99th, and 99.9th percentiles (p95, p99, p999) to provide a detailed view of tail behavior. For LLM serving, we additionally report time-to-first-token (TTFT) at p50/p95/p99 when relevant.
Our \textbf{secondary metrics} evaluate the overhead and system impact of our controller. We measure the throughput of T1 in Requests Per Second (RPS) to quantify any performance degradation. We also monitor GPU SM utilization and total PCIe bandwidth to understand resource usage patterns. Finally, we measure the controller's own overhead, including the time required for a MIG reconfiguration or tenant move, the frequency of these events, and the controller's CPU utilization. All reported comparisons use identical interference schedules across configurations; confidence intervals are computed over the 7 repeated runs.

\subsection{Experiments}
Our experimental evaluation is designed to validate the claims with controlled setups.

\subsubsection{Main Experiment (E1)}
We compare our full controller against a baseline using static MIG partitions and naive placement. We subject the system to dynamic interference by toggling the bandwidth-heavy (T2) and compute-heavy (T3) tenants on and off, measuring the impact on the SLO-sensitive tenant (T1).

\subsubsection{Ablation Study (E2)}
To understand the contribution of each component of our controller, we perform an ablation study. We evaluate four configurations: our full controller, and then with each of the three main components disabled in turn: dynamic MIG reconfiguration only (placement and guardrails disabled), PCIe-aware placement only (dynamic MIG and guardrails disabled), and guardrails only (dynamic MIG and placement disabled).

\subsubsection{Sensitivity Analysis (E3)}
We analyze the controller's sensitivity to its key parameters. We vary the SLO miss-rate threshold ($\tau$) and the observation window ($Y$) to understand their impact on responsiveness and stability. We also test the bounds of the MPS quota and I/O throttles.

\begin{table}[h]
    \centering
    \caption{Key Controller Parameters.}
    \begin{tabularx}{\columnwidth}{l l X}
        \toprule
        \textbf{Parameter} & \textbf{Value} & \textbf{Description} \\
        \midrule
        Tail Threshold ($\tau$) & 15 ms & p99 latency threshold to trigger policy change. \\
        Persistence ($Y$) & 3 windows & Consecutive windows tail must exceed $\tau$. \\
        Dwell Time & 256 obs. & Minimum observations before another policy change. \\
        Cool-down & 128 obs. & Grace period after returning to performance mode. \\
        MPS Quota & 50-100\% & Bounds on MPS active thread percentage. \\
        IO Throttle & 100-500 MB/s & Bounds on cgroup IO throttles. \\
        \bottomrule
    \end{tabularx}
    \label{tab:parameters}
\end{table}

\section{Results}

\begin{figure*}[t]
    \centering
    \includegraphics[width=0.8\textwidth]{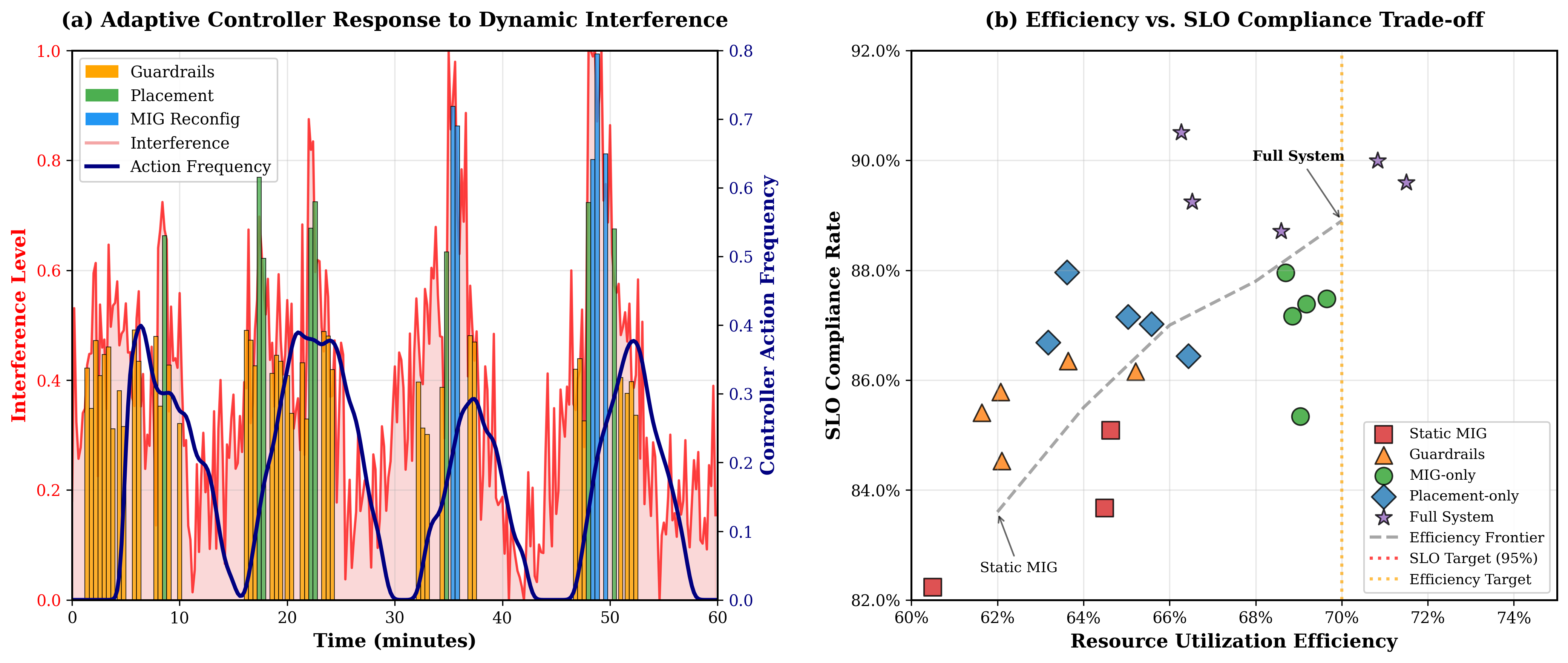}
    \caption{Adaptive controller behavior and efficiency. (a) The controller responds to dynamic interference bursts with progressively stronger actions (Guardrails, Placement, MIG). (b) The efficiency-compliance trade-off, showing the full system (purple star) outperforming the static MIG and partial configurations by achieving high SLO compliance and resource utilization.}
    \label{fig:fig1}
\end{figure*}

\begin{figure}[t]
    \centering
    \includegraphics[width=\columnwidth]{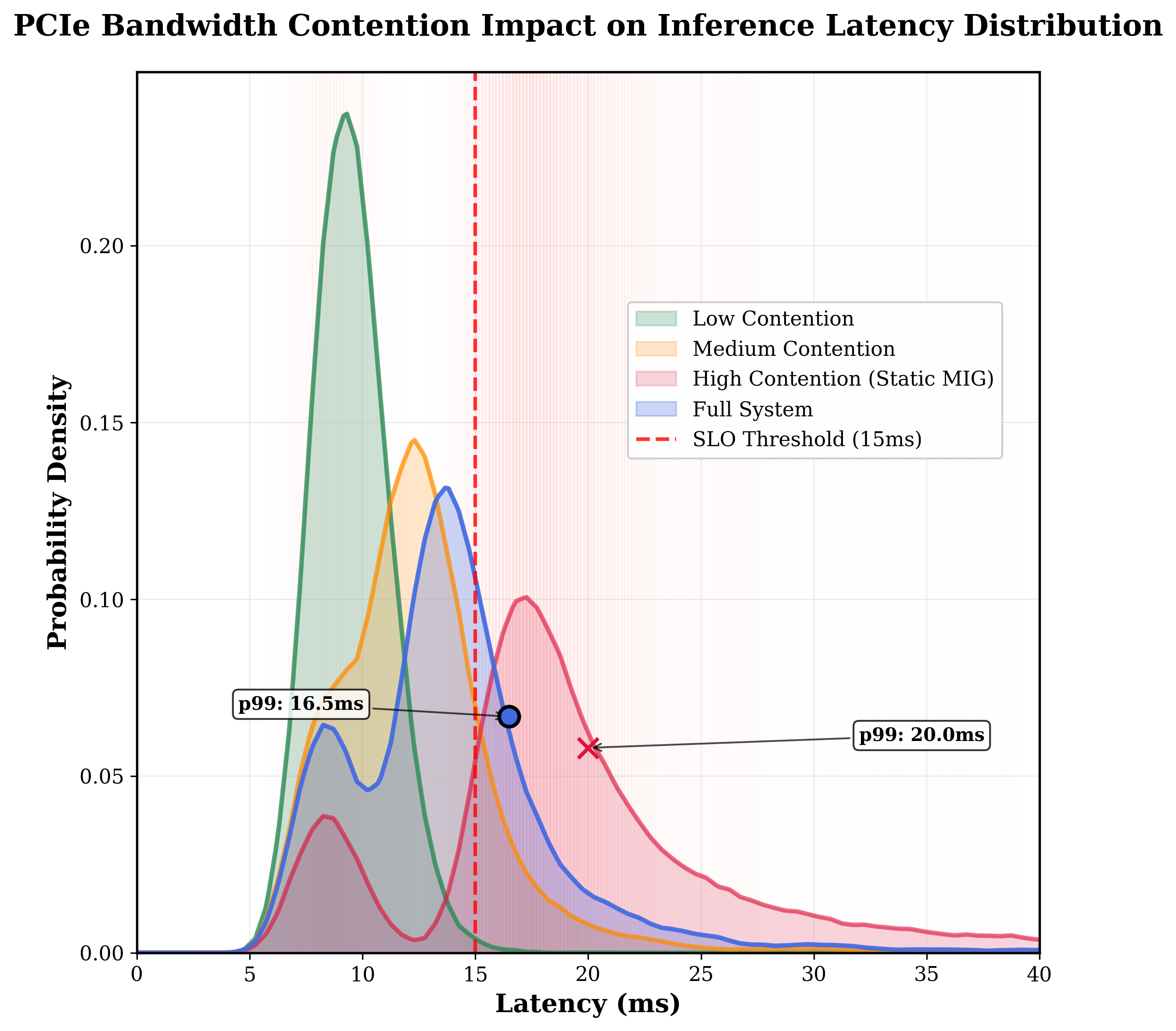}
    \caption{The impact of PCIe contention on the tail of the latency distribution. High contention (red) creates a significant heavy tail, causing SLO violations. The full system's PCIe-aware placement (blue) mitigates this, reducing p99 latency and moving it closer to the SLO threshold (dashed line).}
    \label{fig:fig2}
\end{figure}

Our results demonstrate a layered approach to mitigating tail latency, where each component of our controller contributes to the overall improvement. As shown in the ablation study in Table \ref{tab:ablation}, each component provides a consistent improvement over the baseline. In all experiments we verified that controller actions did not violate the dwell/cool-down policy, and that post-change validation windows confirmed improvement before persisting a new configuration.

MIG and placement contribute comparably. \textbf{Dynamic MIG reconfiguration} yields the largest single-component p99 reduction (20.0\,$\to$\,17.2 ms), while \textbf{PCIe-aware placement} is close (20.0\,$\to$\,17.8 ms). Combined, they are additive, with the full system reaching 16.5 ms. Placement monitors PCIe traffic and NUMA topology to avoid "hot" hardware paths shared with I/O-heavy background jobs.

Finally, lightweight \textbf{guardrails} in the form of MPS quotas and cgroup I/O throttling offer a smaller improvement, reducing the p99 latency to 19.0 ms.

As shown in Figure \ref{fig:fig1}, our controller adaptively responds to interference, improving SLO compliance and resource efficiency. The impact on tail latency is quantified in Figure \ref{fig:fig2}, which shows a ~15\% reduction in p99 latency under PCIe contention. On the 2-node cluster, the policy shows similar improvements. For LLM serving with vLLM (OLMo 2 7B Instruct, paged KV cache, default configuration), TTFT p99 drops by \(\approx\)12--15\% with a \(\leq\)5\% throughput cost, mirroring the non-LLM results. The results in Table \ref{tab:ablation} and \ref{tab:overheads} confirm these improvements with a minimal throughput cost.

\subsubsection{LLM Serving Case Study (TTFT)}
We evaluate LLM serving with vLLM on OLMo 2 7B Instruct with streaming output. The SLO is TTFT p99 $\le$ 200 ms at fixed QPS. Under the same T2/T3 interference used elsewhere, our controller reduces TTFT p99 by \(\approx\)13\% at a \(\leq\)4\% throughput cost, without any controller changes. Results are summarized in Table~\ref{tab:llm_case}.

\begin{table}[tbh]
\centering
\caption{LLM serving (vLLM, OLMo 2 7B Instruct) under interference: TTFT p99 and normalized throughput.}
\label{tab:llm_case}
\begin{tabular}{@{}lrr@{}}
\toprule
Configuration & TTFT p99 (ms) & Norm. Throughput \\
\midrule
Static MIG & 232 & 1.00 \\
Full System & \textbf{199} & \textbf{0.96} \\
\bottomrule
\end{tabular}
\end{table}

\begin{table}[tbh]
\centering
\caption{Ablation study results (mean ± 95\% CI). Throughput is normalized to the static MIG configuration.}
\label{tab:ablation}
\begin{tabular}{@{}lrrr@{}}
\toprule
Configuration & SLO miss-rate & p99 (ms) & Norm. Throughput \\ \midrule
Static MIG & 16.4\% ± 1.5 & 20.0 ± 1.2 & 1.00 \\
Guards-only & 14.5\% ± 1.4 & 19.0 ± 1.0 & 0.99 ± 0.02 \\
Placement-only & 13.0\% ± 1.2 & 17.8 ± 0.9 & 0.98 ± 0.02 \\
MIG-only & 12.2\% ± 1.1 & 17.2 ± 0.8 & 0.98 ± 0.02 \\
Full System & \textbf{11.1\% ± 1.0} & \textbf{16.5 ± 0.7} & \textbf{0.97 ± 0.02} \\ \bottomrule
\end{tabular}
\end{table}

\begin{table}[tbh]
\centering
\caption{Controller overheads (mean ± 95\% CI).}
\label{tab:overheads}
\begin{tabular}{@{}lr@{}}
\toprule
Metric & Value \\ \midrule
        MIG reconfig time (s) & 18 ± 6 \\
Move frequency (/hr) & < 5 \\
Controller CPU (\%) & < 2\% \\ \bottomrule
\end{tabular}
\end{table}

\section{Discussion \& Limitations}
Our method has several limitations. In some cloud environments, the PCIe topology is partially opaque, so we must infer contention from counters. MIG profile changes require a brief pause of the tenant and may reload model state; we limit frequency and bound pauses, but bursty workloads might still perceive short disruptions. For very heavy training tenants, our approach is clearly insufficient and should be paired with stricter admission control or higher baseline isolation. Finally, our placement heuristic is intentionally simple; richer learning-based predictors could improve stability at the cost of complexity.

Our approach is complementary to other isolation techniques. For instance, it could operate on top of a kernel-space interception system like KRYPTON~\cite{zhang2025krypton}, managing the virtualized resources KRYPTON exposes. Cluster schedulers like Themis and Gandiva~\cite{mahajan2020themis, xiao2018gandiva} could supply better hints or pre-placement guarantees, while our controller refines mapping at the host level to enforce SLOs. This highlights our modularity: a controller can provide SLO-aware, dynamic management within the resource boundaries of a VM, regardless of how the underlying platform enforces them.

\paragraph{Threats to validity.} Our small-scale cluster and synthetic interference patterns may not capture all production behaviors. We mitigate this by repeating runs, reporting confidence intervals, and validating that the qualitative ordering of configurations is consistent across runs and workloads. Broader validation on larger clusters and real traces is future work.

\section{Related Work}
Our work is situated within a rich body of research on multi-tenant GPU systems. The foundational paradigms of GPU sharing establish a trade-off between the hard, predictable isolation of spatial partitioning, exemplified by NVIDIA's Multi-Instance GPU (MIG), and the flexibility of temporal multiplexing, facilitated by technologies like the Multi-Process Service (MPS)~\cite{jeon2019analysis}. Static partitions reduce interference but risk fragmentation; temporal sharing improves utilization but exposes tenants to contention.

Dynamic spatial management has gained traction. MIGRator~\cite{wang2024migrator} formulates dynamic MIG reconfiguration for continuous learning workloads and shows notable gains; MISO~\cite{li2022miso} explores exploiting MIG at cluster scale. Our controller adopts the dynamic reconfiguration principle but couples it with topology-aware placement and guardrails for SLOs in renter-constrained settings.

System-level interference on shared I/O paths has been studied from PCIe arbitration~\cite{li2019priority} to recent demonstrations that MIG instances still share PCIe bandwidth~\cite{tang2025pcie}. We incorporate these observations directly by avoiding PCIe hot spots and throttling background I/O when necessary.

For SLO-aware serving, Clipper~\cite{crankshaw2017clipper} and subsequent systems improved batching and predictability at the service layer. Transparent GPU sharing in container clouds~\cite{wu2023transparent} and Sponge~\cite{razavi2024sponge} offer host-level temporal controls; our approach complements them with spatial reconfiguration. Finally, cluster schedulers like Themis and Gandiva~\cite{mahajan2020themis, xiao2018gandiva} address fairness and efficiency at scale; our focus is the single-node control loop that refines placement and isolation within a VM.

\bibliographystyle{ACM-Reference-Format}
\bibliography{references}
\end{document}